\documentclass[10pt]{article}
\usepackage{latexsym,amsthm,amsmath,amssymb,url,textcomp}
\usepackage{amsfonts}
\usepackage{graphicx}
\usepackage{subfigure}
\usepackage[UKenglish]{babel}
\usepackage{algorithm2e}
\usepackage{tikz,comment}
\usetikzlibrary{matrix,arrows,decorations.pathmorphing}

\newtheorem{theorem}{Theorem}
\newtheorem{definition}{Definition}

\newtheorem{corollary}{Corollary}
\newtheorem{lemma}{Lemma}

\newtheorem{rrule}{Reduction Rule}

\newcommand{\smc}{{\sc Signed Max Cut ATLB}}
\newcommand{\mc}{{\sc Max Cut ATLB}}

\newcommand{\interior}[1]{{#1}_{\textsf{int}}}
\newcommand{\exterior}[1]{{#1}_{\textsf{ext}}}
\newcommand{\pt}[1]{\textsf{pt}({#1})}

\title{Maximum Balanced Subgraph Problem Parameterized Above Lower Bound
\footnote{Extended abstract of the paper will appear in the proceedings of COCOON 2013}}
\author
{R. Crowston, G. Gutin, M. Jones and  G. Muciaccia\\Department of Computer
Science\\ Royal Holloway University of London\\ Egham, Surrey TW20
0EX, UK}
\begin{document}

\maketitle

\begin{abstract}
We consider graphs without loops or parallel edges in which every edge
is assigned $+$ or  $-$.
Such a signed graph is balanced if its vertex set can be partitioned
into parts $V_1$ and $V_2$ such that
all edges between vertices in the same part have sign $+$ and
all edges between vertices of different parts have sign $-$ (one
of the parts may be empty). It is well-known that every
connected signed graph with $n$ vertices and $m$ edges has a balanced
subgraph with at least $\frac{m}{2} + \frac{n-1}{4}$ edges and this bound is tight. We
consider the following parameterized problem: given a connected signed
graph $G$  with $n$ vertices and $m$ edges, decide whether $G$ has a
balanced subgraph with at least $\frac{m}{2} + \frac{n-1}{4}+\frac{k}{4}$ edges, where $k$ is
the parameter.

We obtain an algorithm for the problem of runtime $8^k(kn)^{O(1)}$. We
also prove that for each instance $(G,k)$ of the problem, in
polynomial time, we can either solve $(G,k)$ or produce an equivalent
instance $(G',k')$ such that $k'\le k$ and $|V(G')|=O(k^3)$. Our first
result generalizes a result of Crowston, Jones and Mnich (ICALP
2012)  on the corresponding parameterization of Max Cut (when every
edge of $G$ has sign $-$). Our second result generalizes and
significantly improves the corresponding result of Crowston, Jones and
Mnich for MaxCut: they showed that $|V(G')|=O(k^5)$.
\end{abstract}

\section{Introduction}

%In this paper, we consider undirected graphs with no parallel edges or loops and in which every edge is labelled by $=$ or $\neq$. We call such graphs {\em signed graphs}\footnote{Often edges signed graphs  are labelled by the signs $+$ and $-$, but we follow \cite{HufBetNie} and use $=$ and $\neq$ instead.} Let $G=(V,E)$ be a signed graph and let $V=V_1\cup V_2$ be partition of the vertex set of $G$ (i.e., $V_1\cap V_2=\emptyset$). We say that $G$ is {\em $(V_1,V_2)$-balanced} if all edges between vertices of $V_1$ and between the vertices of $V_2$ are labelled by $=$ and all edges between $V_1$ and $V_2$ are labelled by $\neq$; $G$ is {\em balanced} if it is $(V_1,V_2)$-balanced for some partition $V_1,V_2$ of $V$ (the case when one of the sets $V_1$ and $V_2$ is empty is included). A cycle $C$ in $G$ is called {\em positive} ({\em negative}) if the number of edges in $C$ labelled by $\neq$ is even (odd). The following characterization of balanced graphs is well-known.

We consider undirected graphs with no parallel edges or loops and in which every edge is labelled by $+$ or $-$. We call such graphs {\em signed graphs},  and edges, labelled by $+$ and $-$, {\em positive} and {\em negative} edges, respectively. The labels $+$ and $-$ are the {\em signs} of the corresponding edges.
Signed graphs are well-studied due to their various applications and interesting theoretical properties, see, e.g., \cite{CKSXZ,DasEncSonZha,GGMZ,GZ,Harary,HufBetNie,Zas}. 

Let $G=(V,E)$ be a signed graph and let $V=V_1\cup V_2$ be a partition of the vertex set of $G$ (i.e., $V_1\cap V_2=\emptyset$). We say that $G$ is {\em $(V_1,V_2)$-balanced} if an edge with both endpoints in $V_1$, or both endpoints in $V_2$ is positive, and an edge with one endpoint in $V_1$ and one endpoint in $V_2$ is negative; $G$ is {\em balanced} if it is $(V_1,V_2)$-balanced for some partition $V_1,V_2$ of $V$ ($V_1$ or $V_2$ may be empty). 

In some applications, we are interested in finding a maximum-size balanced subgraph of a signed graph \cite{CKSXZ,DasEncSonZha,HufBetNie,Zas}. We will call this problem {\sc Signed Max Cut}. This problem is a generalization of {\sc Max Cut} and as such is NP-hard ({\sc Signed Max Cut} is equivalent to {\sc Max Cut} when all edges of $G$ are negative).   H{\" u}ffner {\em et al.} \cite{HufBetNie} parameterized {\sc Signed Max Cut} below a tight upper bound: decide whether $G=(V,E)$ contains a balanced subgraph with at least $|E|-k$ edges, where $k$ is the parameter\footnote{We use standard terminology on parameterized algorithmics, see, e.g., \cite{DowneyFellows99,FlumGrohe06,Niedermeier06}.}. H{\" u}ffner {\em et al.} \cite{HufBetNie} showed that this parameterized problem is fixed-parameter tractable (FPT) using a simple reduction to the {\sc Edge Bipartization Problem}: decide whether an unsigned graph can be made bipartite by deleting at most $k$ edges ($k$ is the parameter). Using this result and a number of heuristic reductions, H{\" u}ffner {\em et al.} \cite{HufBetNie} designed a nontrivial practical algorithm that allowed them to exactly solve  several instances of {\sc Signed Max Cut}  that were previously solved only approximately by DasGupta {\em et al.} \cite{DasEncSonZha}.

In this paper, we consider a different parameterization of {\sc Signed Max Cut}: decide whether a connected signed graph $G$ with $n$ vertices and $m$ edges contains a subgraph with at least $\frac{m}{2}+\frac{n-1}{4}+\frac{k}{4}$ edges\footnote{We use $\frac{k}{4}$ instead of just $k$ to ensure that $k$ is integral.}, where $k$ is the parameter. Note that $\pt{G}=\frac{m}{2}+\frac{n-1}{4}$ is a tight lower bound on the number of edges in a balanced subgraph of $G$ (this fact was first proved by Poljak and Turz\'{i}k \cite{PoljakTurzik86}, for a different proof, see \cite{CroFelGutJonRosThoYeo}). Thus, we will call this parameterized problem {\sc Signed Max Cut Above Tight Lower Bound} or {\sc Signed Max Cut ATLB}. Whilst the parameterization of H{\" u}ffner {\em et al.} of \mc\ is of interest when the maximum number of edges in a balanced subgraph $H$ of $G$ is close to the number of edges of $G$, {\sc Signed Max Cut ATLB} is of interest when the maximum number of edges in $H$ is close to the minimum possible value in a signed graph on $n$ vertices and $m$ edges. Thus, the two parameterizations treat the opposite parts of the {\sc Signed Max Cut} ``spectrum.''

It appears that it is much harder to prove that {\sc Signed Max Cut ATLB} is FPT than to show that the parameterization of H{\" u}ffner {\em et al.} of {\sc Signed Max Cut} is. Indeed, {\sc Signed Max Cut ATLB} is a generalization of the same parameterization of {\sc Max Cut} (denoted by {\sc Max Cut ATLB}) and the parameterized complexity of the latter was an open problem for many years (and was stated as an open problem in several papers) until Crowston {\em et al.} \cite{CroJonesMnich} developed a new approach for dealing with such parameterized problems\footnote{Recall that {\sc Max Cut} is a special case of {\sc Signed Max Cut} when all edges are negative.}. This approach was applied by Crowston {\em et al.} \cite{CroGutJon} to solve an open problem of Raman and Saurabh \cite{RamSau} on maximum-size acyclic subgraph of an oriented graph. Independently, this problem was also solved by Mnich {\em et al.} \cite{MPSS} who obtained the solution as a consequence of a meta-theorem which shows that several graph problems parameterized above a lower bound of Poljak and Turz\'{i}k \cite{PoljakTurzik86} are FPT under certain conditions. 

While the meta-theorem is for both unlabeled and labeled graphs, all
consequences of the meta-theorem in \cite{MPSS} are proved only for
parameterized problems restricted to unlabelled graphs.
A possible reason is that one of the conditions of the meta-theorem requires us
to show that the problem under consideration is FPT on a special
family of graphs, called almost forests of cliques\footnote{Forests of
cliques are defined in the next section. An almost forest of cliques
is obtained from a forest of cliques by adding to it a small graph
together with some edges linking the small graph with the forest of
cliques.}. The meta-theorem is useful when it is relatively easy to find
an FPT algorithm on almost forests of cliques.
However, for {\sc Signed Max Cut ATLB}
it is not immediately clear what an FPT algorithm would be
even on a clique.

Our attempts to check that {\sc Signed Max Cut ATLB} is FPT
on almost forests of cliques led us to reduction rules that are
applicable not only to almost forests of cliques, but
to arbitrary instances of {\sc Signed Max Cut ATLB}. Thus, we
found two alternatives to prove that {\sc Signed Max Cut ATLB} is FPT:
with and without the meta-theorem. Since the first alternative required
stating the meta-theorem and all related notions and led us to a slightly
slower algorithm than the second alternative,
we decided to use the second alternative.

We reduce an arbitrary instance of {\sc Signed Max Cut ATLB}
to an instance which is an almost forest of cliques, but with an
important additional property which allows us to make use
of a slight modification of a dynamic programming algorithm of
Crowston {\em et al.} \cite{CroJonesMnich} for {\sc Max Cut ATLB} on
almost forests of cliques.

Apart from showing that {\sc Max Cut ATLB} is FPT, Crowston {\em et al.} \cite{CroJonesMnich} proved that the
problem admits a kernel with $O(k^5)$ vertices. They also found
a kernel with $O(k^3)$ vertices for a variation of {\sc Max Cut ATLB},
where the lower bound used is weaker than the Poljak-Turz\'{i}k bound.
They conjectured that a kernel with $O(k^3)$ vertices exists for {\sc
Max Cut ATLB}
as well. In the main result of this paper, we show that {\sc Signed Max Cut ATLB}, which is a more general problem, also admits a polynomial-size
kernel and, moreover, our kernel has $O(k^3)$ vertices. Despite considering a more general
problem than in \cite{CroJonesMnich}, we found a proof which is
shorter and  simpler than the one in \cite{CroJonesMnich}; in
particular, we do not use the probabilistic method. An $O(k^3)$-vertex kernel for {\sc Signed Max Cut ATLB} does not immediately imply  an $O(k^3)$-vertex kernel for {\sc Max Cut ATLB},
but the same argument as for {\sc Signed Max Cut ATLB} shows that {\sc Max Cut ATLB} admits an $O(k^3)$-vertex kernel. This confirms the conjecture above.

\section{Terminology, Notation and Preliminaries}

For a positive integer $l$, $[l]=\{1,\ldots,l\}$.
A cycle $C$ in $G$ is called {\em positive} ({\em negative}) if the number of negative edges in $C$  is even (odd)\footnote{To obtain the sign of $C$ simply compute the product of the signs of its edges.}. The following characterization of balanced graphs is well-known.

\begin{theorem}\cite{Harary}\label{Htheorem}
A signed graph $G$ is balanced if and only if every cycle in $G$ is positive.
\end{theorem}

Let $G=(V,E)$ be a signed graph. For a subset $W$ of $V$, the {\em $W$-switch} of $G$ is the signed graph $G_W$ obtained from $G$ by
changing the signs of the edges between $W$ and $V\setminus W$.
% (note that if $W=V$ or empty, then $G_W=G$). 
Note that a signed graph $G$ is balanced if and only if there exists a subset $W$ of $V$ ($W$ may coincide with $V$) such that $G_W$ has no negative edges. Indeed, if $G_W$ has no negative edges, $G$ is $(W,V\setminus W)$-balanced. If $G$ is $(V_1,V_2)$-balanced, then $G_{V_1}$ has no negative edges. 

Deciding whether a signed graph is balanced is polynomial-time solvable.

\begin{theorem}\label{thm:bal}\cite{GGMZ}
Let $G=(V,E)$ be a signed graph. Deciding whether $G$ is balanced is polynomial-time solvable. Moreover, if $G$ is balanced then, in polynomial time, we can find a subset $W$ of $V$ such that $G_W$ has no negative edges.
\end{theorem}

For a signed graph $G$, $\beta(G)$ will denote the maximum number of edges in a balanced subgraph of $G$. Furthermore, for a signed graph $G=(V,E)$, $\pt{G}$ denotes the Poljak-Turz\'ik bound: $\beta(G)\ge\pt{G}.$ If $G$ is connected, then $\pt{G}=\frac{|E(G)|}{2}+\frac{|V(G)|-1}{4}$, and if $G$ has $t$ components, then $\pt{G}=\frac{|E(G)|}{2}+\frac{|V(G)|-t}{4}$. It is possible to find, in polynomial time, a balanced subgraph of $G$ of size at least $\pt{G}$ \cite{PoljakTurzik86}.

The following easy property will be very useful in later proofs. It follows from Theorem \ref{Htheorem} by observing that for a signed graph the Poljak-Turz\'ik bound does not depend on the signs of the edges and that, for any cycle in $G$, the sign of the cycle in $G$ and in $G_W$ is the same.

\begin{corollary}\label{cor:switch}
Let $G=(V,E)$ be a signed graph and let $W\subset V$. Then $\pt{G_W}=\pt{G}$ and $\beta(G_W)=\beta(G)$.
\end{corollary}

For a vertex set $X$ in a graph $G$, $G[X]$ denotes the subgraph of $G$ induced by $X$. For disjoint vertex sets $X,Y$ of graph $G$, $E(X,Y)$ denotes the
set of edges between $X$ and $Y$. A {\em bridge} in a graph is an edge that, if deleted, increases the number of connected components of the graph. A {\em block} of a graph is either a maximal $2$-connected subgraph or a connected component containing only one vertex.

For an edge set $F$ of a signed graph $G$, $F^+$ and $F^-$ denote the set of positive and negative edges of $F$, respectively.
For a signed graph $G=(V,E)$, the {\em dual} of $G$ is the signed graph $\bar{G}=(V,\bar{E})$, where $\bar{E}^+=E^-$ and $\bar{E}^-=E^+$.
A cycle in $G$ is {\em dually positive} ({\em dually negative}) if the same cycle in $\bar{G}$ is positive (negative).

For a graph $G=(V,E)$, the {\em neighborhood} $N_G(W)$ of $W\subseteq V$ is defined as $\{v\in V:vw\in E,w\in W\}\setminus W$; the vertices in
$N_G(W)$ are called {\em neighbors of $W$}. 
%The {\em closed} neighborhood $N_G[W]$ is defined as $N_G(W)\cup W$.
If $G$ is a signed graph, the {\em positive} neighbors of $W\subseteq V$ are the neighbors of $W$ in $G^+=(V,E^+)$;
the set of positive neighbors is denoted $N_G^+(W)$. Similarly, for the {\em negative} neighbors and $N_G^-(W)$.

The next theorem is the `dual' of Theorem \ref{Htheorem}, in the sense that it is its equivalent formulation on the dual of a graph.

\begin{theorem} \label{nobad}
Let $G=(V,E)$ be a signed graph. Then the dual graph $\bar{G}$ is balanced if and only if $G$ does not contain a dually negative cycle.
\end{theorem}

In the next sections, the notion of {\em forest of cliques} introduced in
\cite{CroJonesMnich} plays a key role. 
A connected graph is a {\em tree of cliques} if the vertices of every cycle induce a clique. A {\em forest of cliques}
is a graph whose components are trees of cliques. It follows from the definition that in a forest of cliques any block is a clique.

Note that a forest of cliques is a {\em chordal} graph, i.e., a graph in which every cycle has a chord, that is an edge between two vertices
which are not adjacent in the cycle.
The next lemma is a characterization of chordal graphs which have a balanced dual.
A {\em triangle} is a cycle with three edges.

\begin{corollary}\label{notriangle}
Let $G=(V,E)$ be a signed chordal graph. Then $\bar{G}$ is balanced if and only if $G$ does not contain a positive triangle.
\end{corollary}
\begin{proof}
If $G$ contains a positive triangle, then, by Theorem \ref{nobad},
$\bar{G}$ is not balanced.

Now suppose that $G$ is not balanced. By Theorem \ref{nobad}, $G$ 
contains a dually negative cycle, i.e., a cycle
with odd number of positive edges, but all triangles in $G$ are
negative by hypothesis. Let $C=v_1v_2\dots v_lv_1$ be a dually negative cycle of
minimum length and note that
$l>3$ as a dually negative triangle is positive. Since the graph is
chordal, we can find three consecutive vertices of $C$ that form a
triangle $T$. Suppose $T=v_1v_2v_3v_1$. Recall that $T$ is negative.
So,
if both $v_1v_2$ and $v_2v_3$ are positive edges (or negative edges),
then $v_1v_3$
must be a negative edge; otherwise if one of the two edges is positive
and the other negative, then $v_1v_3$ is a positive edge. In both
cases, we conclude that
$C$ contains an odd number of positive edges if and only if
$C'=v_1v_3v_4\dots v_lv_1$ does, which is a contradiction since we
supposed $l$ to be the minimum length of a dually negative cycle.
\end{proof}

\begin{corollary}\label{cor:equiv}
Let $(G=(V,E),k)$ be an instance $\cal I$ of \smc, let $X \subseteq
V(G)$ and let $G[X]$ be a chordal graph which does not contain a positive
triangle. Then there exists a set $W \subseteq X$, such that $\widetilde{\cal
I}=(G_W,k)$ is equivalent to $\cal I$, and $G_W[X]$ does not contain
positive edges.
\end{corollary}
\begin{proof}
By Corollary
\ref{notriangle}, $\bar{G}[X]$ is balanced: hence, by definition of
balanced graph,
there exists $W\subseteq X$ such that $\bar{G}_W[X]$ contains only
positive edges, which means that $G_W[X]$ contains only negative edges.
By Corollary \ref{cor:switch}, $(G_W,k)$ is an instance equivalent to the
original one.
\end{proof}

Lastly, the next lemmas describe useful properties of \mc\ which still hold for \smc.

\begin{lemma}\label{half}
Let $G=(V,E)$ be a connected signed graph and let $V=U\cup W$  such that $U\cap W=\emptyset$, $U\neq \emptyset$ and $W\neq \emptyset$. Then $\beta(G)\ge \beta(G[U])+\beta(G[W])+\frac{1}{2}|E(U,W)|$. In addition, if $G[U]$ has $c_1$ components, $G[W]$ has $c_2$ components, $\beta(G[U])\geq\pt{G[U]}+\frac{k_1}{4}$ and $\beta(G[W])\geq\pt{G[W]}+\frac{k_2}{4}$, then $\beta(G)\geq\pt{G}+\frac{k_1+k_2-(c_1+c_2-1)}{4}$.
\end{lemma}
\begin{proof}
Let $H$ ($F$) be a balanced subgraph of $G[U]$ ($G[W]$) with maximum number of edges and let $H$ ($F$) be $(U_1,U_2)$-balanced ($(W_1,W_2)$-balanced). Let 
$E_1=E^+(U_1,W_1)\cup E^+(U_2,W_2)\cup E^-(U_1,W_2)\cup E^-(U_2,W_1)$
and $E_2=E(U,W)\setminus E_1$. Observe that both $E(H)\cup E(F)\cup E_1$ and $E(H)\cup E(F)\cup E_2$ induce balanced subgraphs of $G$ and the largest of them has at least 
$\beta(G[U])+\beta(G[W])+\frac{1}{2}|E(U,W)|$ edges.

Now, observe that $\pt{G}=\pt{G[U]}+\pt{G[W]}+\frac{1}{2}|E(U,W)|+\frac{c_1+c_2-1}{4}$. Hence $\beta(G)\geq\pt{G}+\frac{k_1+k_2-(c_1+c_2-1)}{4}$.
\end{proof}

\begin{lemma}\label{cutvertex}
Let $G=(V,E)$ be a signed graph, $v\in V$ a cutvertex, $Y$ a connected component of $G-v$ and $G'=G-Y$. Then $\pt{G}=\pt{G[V(Y)\cup\{v\}]}+\pt{G'}$ and $\beta(G)=\beta(G[V(Y)\cup\{v\}])+\beta(G')$.
\end{lemma}
\begin{proof}
The first equality is easily verified. Concerning the other, let $H_1$ be a $(V_1^1,V_2^1)$-balanced subgraph of $G[V(Y)\cup\{v\}]$ of size $\beta(G[V(Y)\cup\{v\}])$ and $H_2$ be a $(V_1^2,V_2^2)$-balanced subgraph of $G'$ of size $\beta(G')$. One may assume that $v\in V_1^i$ for $i=1,2$. Therefore the balanced subgraph $H$ of $G$ induced by $V_1=V_1^1\cup V_1^2$ and $V_2=V_2^1\cup V_2^2$ is of size $\beta(G[V(Y)\cup\{v\}])+\beta(G')$, which means that $\beta(G)\geq\beta(G[V(Y)\cup\{v\}])+\beta(G')$.
On the other hand, any balanced subgraph $H$ of $G$ induces balanced subgraphs of $G[V(Y)\cup\{v\}]$ and $G'$, which implies that $\beta(G)\leq\beta(G[V(Y)\cup\{v\}])+\beta(G')$.
\end{proof}

\section{Fixed-Parameter Tractability}\label{sec:fpt}

In this section, we prove that \smc\ is FPT by designing an algorithm of running time\footnote{In the $O^*$-notation widely used in parameterized algorithmics, 
we omit not only constants, but also polynomial factors.}$O^*(8^{k})$.
This algorithm is a generalization of the FPT algorithm obtained in \cite{CroJonesMnich} to solve \mc. 
Given an instance $(G=(V,E),k)$ of \mc, the algorithm presented in \cite{CroJonesMnich} applies some reduction rules that either answer {\sc Yes} for \mc\ or
produce a set $S$ of at most $3k$ vertices such that $G-S$ is a forest of cliques.

A key idea of this section is that it is possible to extend these rules such that we include into $S$ at least one vertex for every dually negative cycle of
$G$.  As a result, Theorem \ref{nobad} ensures that solving \smc\ on $G-S$ is equivalent to solving \mc. Therefore,
it is possible to guess a partial solution on $S$ and then solve {\sc Max-Cut-with-Weighted-Vertices}\footnote{This problem is defined just before Theorem \ref{thm:fpt}.} on $G-S$.
Since a forest of cliques is a chordal graph, Corollary \ref{notriangle} implies that it is enough to put into $S$ at least one vertex for
every positive triangle in $G$ (instead of every dually negative cycle).
Our reduction rules are inspired by the rules used in \cite{CroJonesMnich}, but our rules are more involved in order to deal with positive triangles.

The rules apply to an instance $(G,k)$ of \smc\
and output an instance $(G',k')$ where $G'$ is obtained by deleting some vertices
of $G$. In addition, the rules can mark some of the deleted vertices: marked vertices will form the set $S$ such that $G - S$ is a forest
of cliques. Note that every time a rule marks some vertices, it also decreases the parameter $k$.

The instance $(G',k')$ that the rules produce does not have to be equivalent to $(G,k)$, but it has the property that if it is a {\sc Yes}-instance, then $(G,k)$
is a {\sc Yes}-instance too. For this reason, these rules are called {\em one-way} reduction rules \cite{CroGutJon}.

Note that in the description of the rules, $G$ is a connected signed graph,  and $C$ and $Y$ denote connected components of a signed graph such that $C$ is a clique which does not contain a positive triangle.
%The algorithm applies the reduction rules in the following order, which means that it always applies the %first in the list which can be applied.

\begin{rrule}\label{rule:positivetriangle}\label{rule:first}
%Let $G$ be a connected graph. 
If $abca$ is a positive triangle such that $G-\{a,b,c\}$ is connected, then mark $a,b,c$, delete them
and set $k'=k-3$.
\end{rrule}

\begin{rrule}\label{rule:positivetriangle2}
%Let $G$ be a connected graph. 
If $abca$ is a positive triangle such that $G-\{a,b,c\}$ has two connected components
$C$ and $Y$, 
%and $C$ does not contain a positive triangle, 
then mark $a,b,c$, delete them, delete $C$, and set $k'=k-2$.
\end{rrule}

\begin{rrule}\label{rule:pair}
%Let $G$ be a connected graph. 
Let $C$ be a connected component of $G-v$ for some vertex $v\in G$.
If there exist $a,b\in V(C)$ such that $G-\{a,b\}$ is connected and there is an edge $av$ but no edge $bv$, then mark $a$ and $b$, delete them and set $k'=k-2$.
\end{rrule}

\begin{rrule}\label{rule:pair2}
%Let $G$ be a connected graph. 
Let $C$ be a connected component of $G-v$ for some vertex $v\in G$.
If there exist $a,b\in C$ such that $G-\{a,b\}$ is connected and $vabv$ is a positive triangle, then
mark $a$ and $b$, delete them and set $k'=k-4$.
\end{rrule}

\begin{rrule}\label{rule:clique}
%Let $G$ be a connected graph. 
If there is a vertex $v\in V(G)$ such that
$G-v$ has a connected component $C$,
$G[V(C)\cup\{v\}]$ is a clique in $G$,
and $G[V(C)\cup\{v\}]$ does not contain a positive triangle,
then delete $C$ and set $k'=k$.
\end{rrule}

\begin{rrule}\label{rule:path}
%Let $G$ be a connected graph. 
If $a,b,c\in V(G)$, $\{ab,bc\}\subseteq E(G)$ but $ac\notin E(G)$, and $G-\{a,b,c\}$ is connected,
then mark $a,b,c$, delete them and set $k'=k-1$.
\end{rrule}

\begin{rrule}\label{rule:pincer}\label{rule:last}
%Let $G$ be a connected graph. 
Let $C,Y$ be the connected components of
$G-\{v,b\}$ for some vertices $v,b\in V(G)$ such that $vb\notin E(G)$. If $G[V(C)\cup\{v\}]$ and $G[V(C)\cup\{b\}]$ are cliques not containing any positive triangles,
then mark $v$ and $b$, delete them, delete $C$ and set $k'=k-1$.
\end{rrule}

\begin{definition}
A one-way reduction rule is {\em safe} if it does not transform a {\sc No}-instance into a {\sc Yes}-instance.
\end{definition}

% One-way reduction rules are required to be safe.
The intuitive understanding of how a one-way reduction rule works is that it removes
a portion of the graph (while decreasing the parameter from $k$ to $k'$) only if given any solution (i.e., a balanced subgraph) on the rest of the graph
there is a way to extend it to the removed portion while always gaining an additional $k-k'$ over the Poljak-Turz\'ik bound.

%Before proving that Reduction Rules \ref{rule:first}-\ref{rule:last} are safe, we show the following:

\begin{lemma}\label{notriangleincliques}
Let $G$ be a connected graph. If $C$ is a clique of $G$ such that $G-C$ is connected and if $C$ contains a positive triangle, then either Rule \ref{rule:positivetriangle}
 or Rule \ref{rule:positivetriangle2} applies.
\end{lemma}
\begin{proof}
Let $abca$ be a positive triangle in $C$. Suppose Rule \ref{rule:positivetriangle} does not apply. This means that $G-\{a,b,c\}$ is not connected: more precisely, $G-\{a,b,c\}$ has two components $G-C$ and $C-\{a,b,c\}$. Note that $C-\{a,b,c\}$ cannot contain a positive triangle, or otherwise Rule \ref{rule:positivetriangle} would have applied. Therefore, Rule \ref{rule:positivetriangle2} applies.
\end{proof}

%The next theorem shows that Reduction Rules \ref{rule:first}-\ref{rule:last} are safe. 
%Notice that a positive triangle is a balanced graph and,
%since the Poljak-Turz\'ik bound on a triangle is $\frac{3}{2}+\frac{3-1}{4}=2$, 
%on a positive triangle it is possible to gain
%$1$ over the bound. Recall that $(G,k)$ denotes 
%the instance on which one of the rules applies and $(G',k')$ denotes the instance which is produced.

\begin{theorem}\label{lem-safe:positivetriangle}
Rules \ref{rule:first}-\ref{rule:last} are safe.
\end{theorem}
\begin{proof}

\noindent{\bf Rule \ref{rule:positivetriangle}:} Let $abca$ be a positive triangle  as in the description of Rule \ref{rule:positivetriangle}. Suppose $\beta(G')\geq\pt{G'}+\frac{k'}{4}$, where $k'=k-3$.
Since $abca$ is a positive triangle, by Lemma \ref{half}, we obtain $\beta(G)\geq\pt{G}+\frac{k'+3}{4}=\pt{G}+\frac{k}{4}$.

\smallskip 

\noindent{\bf Rule \ref{rule:positivetriangle2}:} Let $abca$ be a positive triangle such that $G-\{a,b,c\}$ has two components $C$ and $Y$. Suppose $\beta(Y)\geq\pt{Y}+\frac{k'}{4}$, where $k'=k-2$.
We know that $\beta(C)\geq\pt{C}$, and so by Lemma \ref{half} we obtain $\beta(G[V(Y)\cup V(C)])\geq\pt{G[V(Y)\cup V(C)]}+\frac{k'-1}{4}$.
Since $abca$ is a positive triangle, using Lemma \ref{half} again we obtain $\beta(G)\geq\pt{G}+\frac{k'+4-2}{4}=\pt{G}+\frac{k}{4}$.

\smallskip 

\noindent{\bf Rule \ref{rule:pair}:}
Let $v,a,b$ and $C$ be as in the description of Rule \ref{rule:pair}. Assume there exists a $(V_1',V_2')$-balanced subgraph $H'$ of $G'$ with at least $\pt{G'}+\frac{k'}{4}$ edges, where $k'=k-2$.
%Note that by Lemma \ref{notriangleincliques}, if $C$ contained a positive triangle, Rule %\ref{rule:positivetriangle} or Rule \ref{rule:positivetriangle2} would apply. Hence $C$ does not contain a %positive triangle and, 
By Corollary \ref{cor:equiv}, we may assume that all edges in $C$ are negative. In addition, we may assume that the edge $av$ is negative (the other case is similar). Lastly, without loss of generality assume that $v\in V_1'$.
Now, consider the balanced subgraph $H$ of $G$ induced by $(V_1,V_2)$, where $V_1=V_1'\cup\{b\}$ and $V_2=V_2'\cup\{a\}$. Since $|E(a,V_1\cap V(C))\cup E(b,V_2\cap V(C))|=|E(a,V_2\cap V(C))\cup E(b,V_1\cap V(C))|$, it holds that $|E(H)|=|E(H')|+\frac{|E(\{a,b\},V[C-\{a,b\}])|}{2}+2$. Moreover, $\pt{G}=\pt{G'}+\frac{|E(\{a,b\},V[C-\{a,b\}])|}{2}+\frac{3}{2}$. Thus, $\beta(G)\geq\pt{G}+\frac{k'+2}{4}=\pt{G}+\frac{k}{4}$.

\smallskip 

\noindent{\bf Rule \ref{rule:pair2}:}
Let $v,a,b$ and $C$ be as in the description of Rule \ref{rule:pair2}. Assume there exists a $(V_1',V_2')$-balanced subgraph $H'$ of $G'$ with at least $\pt{G'}+\frac{k'}{4}$ edges, where $k'=k-4$.
As in the proof for Rule \ref{rule:pair}, assume $C$ only contains negative edges, the edge $av$ is negative and $v\in V_1'$. Since $vabv$ is a positive triangle, this implies that the edge $bv$ is positive.

Now, consider the balanced subgraph $H$ of $G$ induced by $(V_1,V_2)$, where $V_1=V_1'\cup\{b\}$ and $V_2=V_2'\cup\{a\}$. As in the proof for Rule \ref{rule:pair},
it holds that $|E(a,V_1\cap V(C))\cup E(b,V_2\cap V(C))|=|E(a,V_2\cap V(C))\cup E(b,V_1\cap V(C))|$. Hence, $|E(H)|=|E(H')|+\frac{|E(\{a,b\},V[C-\{a,b\}])|}{2}+3$, while $\pt{G}=\pt{G'}+\frac{|E(\{a,b\},V[C-\{a,b\}])|}{2}+2$. Thus, $\beta(G)\geq\pt{G}+\frac{k'+4}{4}=\pt{G}+\frac{k}{4}$.

\smallskip 

\noindent{\bf Rule \ref{rule:clique}:} Let $v$ and $C$ be as in the description of Rule \ref{rule:clique}. Suppose $\beta(G')\geq\pt{G'}+\frac{k}{4}$. We know that $\beta(G[V(C)\cup\{v\}])\geq\pt{G[V(C)\cup\{v\}]}$. Then, by Lemma \ref{cutvertex}, $\beta(G)\geq\pt{G}+\frac{k}{4}$.

\smallskip 

\noindent{\bf  Rule \ref{rule:path}:} Let $a,b,c$ be as in the description of Rule \ref{rule:path} and let $P=G[\{a,b,c\}].$ Note that $\pt{P}=\frac{2}{2}+\frac{3-1}{4}=\frac{3}{2}$ and,
whatever the signs of its edges, $P$ is a balanced graph by Theorem \ref{Htheorem}.
Therefore, $\beta(P)=2=\pt{P}+\frac{1}{2}$.
Suppose $\beta(G')\geq\pt{G'}+\frac{k'}{4}$, where $k'=k-1$. Then by Lemma \ref{half}, $\beta(G)\geq\pt{G}+\frac{k'+2-1}{4}=\pt{G}+\frac{k}{4}$ edges.

\smallskip 

\noindent{\bf  Rule \ref{rule:pincer}:} Let $v,b$ and $C,Y$ be as in the description of Rule \ref{rule:pincer}.
Suppose $\beta(Y)\ge \pt{Y}+\frac{k'}{4}$, where $k'=k-1$. We claim that $\beta(G-Y)\geq \pt{G-Y}+\frac{1}{2}$. If this holds, using Lemma \ref{half} we obtain that $\beta(G)\geq\pt{G}+\frac{k'+1}{4}=\pt{G}+\frac{k}{4}$.

Let $G''=G-Y$ and $V(C)=\{v_1,\dots,v_n\}$. Note that $\pt{G''}=\frac{n(n-1)}{4}+n+\frac{n+1}{4}$.

%We claim that $G''$ cannot contain a positive triangle. If the claim holds, since $G''$ is chordal, 
By Corollary \ref{cor:equiv} we may assume that $G''$ contains only negative edges. If $n$ is even, consider the partition $(V_1,V_2)$ of $V(G'')$ where $V_1=\{v,b,v_1,\dots,v_{\frac{n}{2}-1}\}$ and $V_2=\{v_{\frac{n}{2}},\dots,v_n\}$. The balanced subgraph induced by this partition contains $(\frac{n}{2}+1)^2=\pt{G''}+\frac{3}{4}$ edges. On the other hand, if $n$ is odd, consider the partition $(V_1,V_2)$ of $V(G'')$ where $V_1=\{v,b,v_1,\dots,v_{\frac{n-1}{2}}\}$ and $V_2=\{v_{\frac{n+1}{2}},\dots,v_n\}$. The balanced subgraph induced by this partition contains $\frac{n+3}{2}\cdot\frac{n+1}{2}= \pt{G''}+\frac{2}{4}$ edges.
%Now we will prove the claim. For the sake of contradiction, suppose that
%$G''$ contains a positive triangle. Since $G$ is connected, there is an
%edge between $Y$ and $v$ or $b$ (or both). We will show that there cannot
%be edges between $Y$ and both $v$ and $b$. In fact, if it were the case,
%by Lemma \ref{notriangleincliques} either Rule \ref{rule:positivetriangle}
%or \ref{rule:positivetriangle2} would have applied because both
%$G-(V(C)\cup\{v\}) $ and $G-(V(C)\cup\{b\})$ are connected. Hence, without loss
%of generality, suppose that there is an edge between $v$ and $Y$ and no
%edges between $b$ and $Y$. Then Rule \ref{rule:pair} applies, which is a
%contradiction.
\end{proof}

We now show that the reduction rules preserve connectedness and that there is always one of them which applies to a graph with at least one edge. To show this, we use the following lemma, based on a result in \cite{CroJonesMnich} but first expressed in the following form in \cite{CroGutJon}.

\begin{lemma}\cite{CroGutJon}\label{lem:slem}
For any connected graph $Q$, at least one of the following properties holds:
\begin{description}
\item[A] There exist $v \in V(Q)$ and $X\subseteq V(Q)$ such that $G[X]$ is a connected component of $Q-v$ and $G[X]$ is a clique;
\item[B] There exist $a,b,c \in V(Q)$ such that $Q[\{a,b,c\}]$ is isomorphic to path $P_3$ and $Q-\{a,b,c\}$ is connected;
\item[C] There exist $v,b\in V(Q)$ such that $vb \notin E(Q)$, $Q-\{v,b\}$ is disconnected, and for all connected components $G[X]$ of $Q-\{v,b\}$, except possibly one, $G[X \cup \{v\}]$ and $G[X \cup \{b\}]$ are cliques.
\end{description}
\end{lemma}

\begin{lemma}\label{lem:exhaust}
For a connected graph $G$ with at least one edge, at least one of Rules \ref{rule:first}-\ref{rule:last} applies. In addition, the graph $G'$
which is produced is connected.
\end{lemma}
\begin{proof}
It is not difficult to see that the graph $G'$ is connected, since when it is not obvious, its connectedness is part of the conditions
for the rule to apply.

If property A of Lemma \ref{lem:slem} holds, and $G[X]$ contains a positive triangle $abca$, then by Lemma \ref{notriangleincliques} either Rule \ref{rule:positivetriangle} or Rule \ref{rule:positivetriangle2} applies. If $2\le |N_G(v)\cap X|\le |X|-1$, then Rule \ref{rule:pair} applies. If $|N_G(v)\cap X|=|X|$ and there exist $a,b\in X$ such that $vabv$ is a positive triangle, Rule \ref{rule:pair2} applies; otherwise, $G[X\cup \{v\}]$ contains no positive triangles, and Rule \ref{rule:clique} applies. Finally, if $N_G(v)\cap X=\{x\}$, Rule \ref{rule:clique} applies for $x$ with clique $G[X\setminus \{x\}]$.

If property B of Lemma \ref{lem:slem} holds, then Rule \ref{rule:path} applies. If property C of Lemma \ref{lem:slem} holds, consider the case when $G-\{v,b\}$ has two connected components. Let $Z$ be the other connected component. If $Z$ is connected to only
one of $v$ or $b$, then property A holds. Otherwise, if $G[X\cup
\{x\}]$ contains a positive triangle, where $x\in \{v,b\}$, then by Lemma \ref{notriangleincliques}  either Rule \ref{rule:positivetriangle} or Rule \ref{rule:positivetriangle2} applies. So we may assume that $G[X\cup \{b,v\}]$ contains no positive
triangles, in which case Rule \ref{rule:pincer}  applies.

If $G-\{v,b\}$ has at least three connected components, at least two of them, $X_1,X_2$, form cliques with both $v$ and $b$ and possibly one component $Y$ does not. Assume without loss of generality that $Y$ has an edge to $v$. Then Rule \ref{rule:path} applies for the path $x_1bx_2$, where $x_1\in X_1$, $x_2\in X_2$.

\end{proof}

The following lemma gives structural results on $S$ and $G-S$. Note that from now on, $(G=(V,E),k)$ denotes the original instance of \smc\
and $(G'=(V',E'),k')$ denotes the instance obtained by applying Rules \ref{rule:first}-\ref{rule:last} exhaustively. The set $S\subseteq V$
denotes the set of vertices which are marked by the rules.

\begin{lemma}\label{lemma:forest}
Given a connected graph $G$, if we apply Rules \ref{rule:first}-\ref{rule:last} exhaustively, either
the set $S$ of marked vertices has cardinality at most $3k$, or $k'\leq 0$. In addition, $G-S$ is a forest of cliques that does not contain a positive triangle.
\end{lemma}
\begin{proof}
%If we consider the ratio of the number of marked vertices over $k-k'$ 
Observe that for every reduction rule where some vertices are marked, at most 3 vertices are marked, and the parameter descreases by at least 1.
%we see that it achieves the biggest value in the case of Rule \ref{rule:path}, where it is equal to $3$. 
This means that if $k'>0$, then the reduction rules cannot
have marked more than $3k$ vertices.

To show that $G-S$ is a forest of cliques that does not contain a positive triangle, proceed by induction. It is trivially true that the empty graph and the graph
with only one vertex are forests of cliques that do not contain positive triangles. Now suppose that we apply one of the rules, transforming a graph $G_1$ into a graph $G_2$;
suppose in addition that $G_2-(S\cap V(G_2))$ is a forest of cliques that does not contain a positive triangle: we claim
that $G_1-(S\cap V(G_1))$ is a forest of cliques that does not contain a positive triangle, too.
In the case of Rules \ref{rule:positivetriangle}, \ref{rule:pair}, \ref{rule:pair2} and \ref{rule:path}, $G_1-(S\cap V(G_1))$ is equal to $G_2-(S\cap V(G_2))$,
therefore the claim is trivially true.
For Rule \ref{rule:clique}, note that $G_1-(S\cap V(G_1))$ is obtained from $G_2-(S\cap V(G_2))$ by either adding a disjoint
clique not containing a positive triangle if $v\in S$, or adding a clique not containing a positive triangle and identifying one of its vertices with $v$ (where $v$ is a cutvertex as in the description of
Rule \ref{rule:clique}).
Finally, for Rules \ref{rule:positivetriangle2} and \ref{rule:pincer}, $G_1-(S\cap V(G_1))$ is obtained from $G_2-(S\cap V(G_2))$ by
adding one disjoint clique not containing a positive triangle.
\end{proof}

%Now we are able to show what we claimed in the beginning of this section.

%\begin{lemma}\label{lemma:notriangle}[$\star$]
%$G-S$ does not contain positive triangles.
%\end{lemma}

Finally, it is possible to prove that \smc\ is FPT. First we state \textsc{Max-Cut-with-Weighted-Vertices} as in \cite{CroJonesMnich}.

\begin{center}
\fbox{~\begin{minipage}{0.8\textwidth}
\textsc{Max-Cut-with-Weighted-Vertices}

\emph{Instance}: A graph $G$ with weight functions
$w_1:~V(G)~\rightarrow~\mathbb{N}_0$ and
$w_2:~V(G)~\rightarrow~\mathbb{N}_0$, and an integer $t\in\mathbb N$.

\smallskip

\emph{Question}: Does there exist an assignment $f: V(G) \rightarrow \{1,2\}$
such that
$\sum_{xy \in E}|f(x)-f(y)| + \sum_{f(x)=1}w_1(x) + \sum_{f(x)=2}w_2(x)
\ge t$?
\end{minipage}~}
\end{center}

\begin{theorem}\label{thm:fpt}
\smc\ can be solved in time $O^*(8^k)$.
\end{theorem}
\begin{proof}
Let $(G=(V,E),k)$ be an instance of \smc. Apply Rules \ref{rule:first}-\ref{rule:last} exhaustively, producing an instance $(G'=(V',E'),k')$
and a set $S\subseteq V$ of marked vertices. If $k'\leq 0$, $(G',k')$ is a trivial {\sc Yes}-instance. Since the rules are safe, it follows that
$(G,k)$ is a {\sc Yes}-instance, too.

Otherwise, $k'>0$. Note that by Lemma \ref{lemma:forest}, $|S|\le 3k$ and $G-S$ is a forest of cliques, which is
a chordal graph without positive triangles. Hence, by Corollary \ref{cor:equiv},
we may assume that $G-S$ does not contain positive edges.

Therefore, to solve \smc\ on $G$, we can guess a balanced subgraph of $G[S]$, induced by a partition $(V_1,V_2)$, and then solve 
{\sc Max-Cut-with-Weighted-Vertices} for $G-S$. The weight of a vertex $v\in V(G-S)$ is defined
in the following way: let $n_i^+(v)$ be the number of positive neighbors of $v$ in $V_i$ and $n_i^-(v)$ be the number of negative neighbors
of $v$ in $V_i$; then $w_1(v)=n_1^+(v)+n_2^-(v)$ and $w_2(v)=n_2^+(v)+n_1^-(v)$.

Since {\sc Max-Cut-with-Weighted-Vertices} is solvable in polynomial time on a forest of cliques (see Lemma 9 in \cite{CroJonesMnich}) and
the number of possible partitions of $S$ is bounded by $2^{3k}$, this gives an $O^*(8^k)$-algorithm to solve \smc.
\end{proof}

\section{Kernelization}
In this section, we show that \smc\ admits a kernel with $O(k^3)$ vertices. The proof of Theorem \ref{thm:fpt} implies the following key result for our kernelization.

\begin{corollary}
Let $(G=(V,E),k)$ be an instance of \smc. In polynomial time, either we can conclude that $(G,k)$ is a {\sc Yes}-instance or we can find a set $S$ of at most $3k$ vertices for which we may assume that $G-S$ is a forest of cliques without positive edges. 
\end{corollary}

The kernel is obtained via the application of a new set of reduction rules and using structural results that bound the size of {\sc No}-instances $(G,k)$. First, we need some additional terminology.
For a block $C$ in $G-S$, let $\interior{C}=\{x\in V(C):N_{G-S}(x)\subseteq V(C)\}$ be the {\em interior} of $C$, and let $\exterior{C}=V(C)\setminus \interior{C}$ be the {\em exterior} of $C$.
If a block $C$ is such that $\interior{C}\cap N_G(S)\neq \emptyset$, $C$ is a {\em special} block. We say a block $C$ is a {\em path block} if $|V(C)|=2=|\exterior{C}|$. A {\em path vertex} is a vertex which is contained only in path blocks. A block $C$ in $G-S$ is a {\em leaf block} if $|\exterior{C}|\leq 1$.

%\subsection{Reduction Rules}
The following reduction rules are two-way reduction rules: they apply to an instance $(G,k)$ and produce an equivalent instance $(G',k')$.

\begin{rrule}\label{kerRuleRemove2}\label{kerRuleFirst}
Let $C$ be a block in $G-S$. If there exists $X \subseteq \interior{C}$ such that $|X| > \frac{|V(C)|+|N_G(X)\cap S|}{2}\ge 1$, $N_G^+(x)\cap S=N_G^+(X)\cap S$ and $N_G^-(x)\cap S=N_G^-(X)\cap S$ for all $x\in X$, then delete two arbitrary vertices $x_1,x_2\in X$ and set $k'=k$.
\end{rrule}

\begin{rrule}\label{kerRuleRemove1}
Let $C$ be a block in $G-S$. If $|V(C)|$ is even and there exists $X \subseteq \interior{C}$ such that $|X|=\frac{|V(C)|}{2}$ and $N_G(X)\cap S=\emptyset$, then delete a vertex $x\in X$ and set $k'=k-1$.
\end{rrule}

\begin{rrule}\label{kerRsize3see2}
Let $C$ be a block in $G-S$ with vertex set $\{x,y,u\},$ such that $N_G(u)=\{x,y\}$.
If the edge $xy$ is a bridge in $G-\{u\}$, delete $C$, add a new vertex $z$, positive edges 
$\{ zv : v\in N_{G-u}^+(\{x,y\}) \}$, negative edges $\{ zv : v\in N_{G-u}^-(\{x,y\}) \}$ and set $k'=k$.
Otherwise, delete $u$ and the edge $xy$ and set $k'=k-1$.
\end{rrule}

\begin{rrule}\label{kerRuleOneinS}\label{kerRuleLast}
Let $T$ be a connected component of $G-S$ only adjacent to a vertex $s\in S$. Form a
{\sc Max-Cut-with-Weighted-Vertices} instance on $T$ by defining $w_1(x)=1$ if $x\in N_G^+(s)\cap T$ ($w_1(x)=0$ otherwise) and $w_2(y)=1$ if $y\in N_G^-(s)\cap T$ ($w_2(y)=0$ otherwise). Let $\beta(G[V(T)\cup\{s\}])=\pt{G[V(T)\cup\{s\}]}+\frac{p}{4}$. Then delete $T$ and set $k'=k-p$. \end{rrule}

Note that the value of $p$ in Rule \ref{kerRuleLast} can be found in polynomial time by solving {\sc Max-Cut-with-Weighted-Vertices} on $T$.
%for all possible values of $t.$

A two-way reduction rule is {\em valid} if it transforms {\sc Yes}-instances into {\sc Yes}-instances and {\sc No}-instances into {\sc No}-instances.
Theorem \ref{lem:kerRuleRemove2valid} shows that Rules \ref{kerRuleFirst}-\ref{kerRuleLast} are valid.  To prove Theorem \ref{lem:kerRuleRemove2valid}, we need the following two lemmas.

\begin{lemma}\label{lem:partsizes}
Let $C$ be a block in $G-S$. If there exists $X\subseteq\interior{C}$ such that $|X|\geq\frac{|V(C)|}{2}$, then there exists a $(V_1,V_2)$-balanced subgraph $H$ of $G$ with $\beta(G)$ edges
such that at least one of the following inequalities holds:
\begin{itemize}
\item $|V_2\cap V(C)|\leq |V_1\cap V(C)|\leq |N_G(X)\cap S|+|V_2\cap V(C)|$;
\item $|V_2\cap V(C)|\leq |V_1\cap V(C)|\leq |V_2\cap V(C)|+1$.
\end{itemize}
%Moreover, if $N_G(X)\cap S=\emptyset$ then the second inequality holds.
\end{lemma}

\begin{proof}
We may assume that $|V_1\cap V(C)|\ge |V_2\cap V(C)|$. Note that if $|V_1\cap V(C)|>|V_2\cap V(C)|$, then $X\cap V_1\neq\emptyset$ (because $|X|\geq\frac{|V(C)|}{2}$).

First, if $N_G(X)\cap S=\emptyset$ and $|V_1\cap V(C)|\ge |V_2\cap V(C)|+2$, then, for any $x\in X\cap V_1$, the subgraph induced by the partition $(V_1\setminus\{x\},V_2\cup\{x\})$ has more edges than the subgraph induced by $(V_1,V_2)$, which is a contradiction.

Now, suppose that $N_G(X)\cap S\neq\emptyset$ and suppose also that $|V_1\cap V(C)|-|V_2\cap V(C)|$ is minimal. If $|V_1\cap V(C)|\leq |V_2\cap V(C)|+1$ we are done, so suppose $|V_1\cap V(C)|\geq |V_2\cap V(C)|+2$. Consider the partition $V_1'=V_1\setminus \{x\}$, $V_2'=V_2 \cup \{ x\}$, where $x\in V_1 \cap X$, and the balanced subgraph $H'$ induced by this partition. Then $|E(H')|\geq |E(H)|+|E(V_1\setminus \{x\},x)|-|E(V_2,x)|\geq |E(H)|+(|V_1\cap V(C)|-1-|N_G(X)\cap S|-|V_2\cap V(C)|)$.
Since $|V_1'\cap V(C)|-|V_2'\cap V(C)|<|V_1\cap V(C)|-|V_2\cap V(C)|$, it holds that $|E(H')|\leq |E(H)|-1$. Therefore, $|V_1\cap V(C)|\leq |N_G(X)\cap S|+|V_2\cap V(C)|$.
\end{proof}

% \noindent{\bf Lemma \ref{lem:xsplit}.}
% {\em Let $C$ be a block in $G-S$. If there exists $X\subseteq\interior{C}$ such that $|X| > \frac{|V(C)|+|N_G(X)\cap S|}{2}$, $N_G^+(x)\cap S=N_G^+(X)\cap S$ and $N_G^-(x)\cap S=N_G^-(X)\cap S$ for all $x\in X$, then, for any $x_1,x_2\in X$, there exists a $(V_1,V_2)$-balanced subgraph $H$ of $G$ with $\beta(G)$ edges such that $x_1\in V_1$ and $x_2\in V_2$.}
% 
% 

\begin{lemma}\label{lem:xsplit}
Let $C$ be a block in $G-S$. If there exists $X\subseteq\interior{C}$ such that $|X| > \frac{|V(C)|+|N_G(X)\cap S|}{2}$, $N_G^+(x)\cap S=N_G^+(X)\cap S$ and $N_G^-(x)\cap S=N_G^-(X)\cap S$ for all $x\in X$, then, for any $x_1,x_2\in X$, there exists a $(V_1,V_2)$-balanced subgraph $H$ of $G$ with $\beta(G)$ edges such that $x_1\in V_1$ and $x_2\in V_2$.
\end{lemma}

\begin{proof}
First, we claim that there exist vertices $x_1,x_2\in X$ for which the result holds. Let $H$ be a $(V_1,V_2)$-balanced subgraph of $G$ with $\beta(G)$ edges as given by Lemma \ref{lem:partsizes}.

Suppose $N_G(X)\cap S=\emptyset$. Then, by Lemma \ref{lem:partsizes} it holds that $|V_2\cap V(C)|\leq |V_1\cap V(C)|\leq |V_2\cap V(C)|+1$; in addition, $|X|>\frac{|V(C)|}{2}$. Hence, either we can find $x_1$ and $x_2$ as required, or $X=V_1\cap V(C)$ and $|V_1\cap V(C)|=|V_2\cap V(C)|+1$. In the second case, pick a vertex $x\in V_1$ and form the partition $V_1'=V_1\setminus\{x\}$ and $V_2'=V_2\cup\{x\}$. Consider the balanced subgraph $H'$ induced by this partition. Observe that $|E(H')|=|E(H)|-|E(x,V_2)|+|E(x,V_1\setminus\{x\})|=|E(H)|-|V_2\cap V(C)|+|V_1\cap V(C)|-1=|E(H)|$, so $H'$ is a maximum balanced subgraph for which we can find $x_1$ and $x_2$ as required.

Now, suppose $N_G(X)\cap S\neq\emptyset$. Then by Lemma \ref{lem:partsizes} it holds that $|V_2\cap V(C)|\leq |V_1\cap V(C)|\leq |N_G(X)\cap S|+|V_2\cap V(C)|$. For the sake of contradiction, suppose $X\subseteq V_1\cap V(C)$ or $X\subseteq V_2\cap V(C)$: in both cases, this means that $|X|\leq |V_1\cap V(C)|$. Note that $|V(C)|=|V_1\cap V(C)|+|V_2\cap V(C)|=2|V_2\cap V(C)|+t$, where $t\leq |N_G(X)\cap S|$. Hence, $|V_1\cap V(C)|\geq |X|>\frac{|V(C)|+|N_G(X)\cap S|}{2}=|V_2\cap V(C)|+\frac{t}{2}+\frac{|N_G(X)\cap S|}{2}\geq |V_2\cap V(C)|+t=|V_1\cap V(C)|$, which is a contradiction.

To conclude the proof, notice that for a $(V_1,V_2)$-balanced subgraph $H$ of $G$ with $\beta(G)$ edges and vertices $x_1,x_2\in X$ such that $x_1\in V_1$ and $x_2\in V_2$, we have  $|E(H)|=|E(H')|$, where $H'$ is a balanced subgraph induced by $V_1'=V_1\setminus\{x_1\}\cup\{x_2\}$ and $V_2'=V_2\setminus\{x_2\}\cup\{x_1\}$: this is true because $N_G^+(x_1)\cap S=N_G^+(x_2)\cap S$ and $N_G^-(x_1)\cap S=N_G^-(x_2)\cap S$.
\end{proof}

\begin{theorem}\label{lem:kerRuleRemove2valid}
Rules \ref{kerRuleRemove2}-\ref{kerRuleOneinS} are valid.
\end{theorem}
\begin{proof}

\noindent{\bf Rule \ref{kerRuleRemove2}:} Let $C,X$ be as in the description of Rule \ref{kerRuleRemove2}. Let $x_1,x_2\in X$. By Lemma \ref{lem:xsplit}, there exists a $(V_1,V_2)$-balanced subgraph $H$ of $G$ with $\beta(G)$ edges such that $x_1\in V_1$ and $x_2\in V_2$. Now, let $G'=G-\{x_1,x_2\}$ and $H'=H-\{x_1,x_2\}$. Since $N_G^+(x_1)\cap S=N_G^+(x_2)\cap S$ and $N_G^-(x_1)\cap S=N_G^-(x_2)\cap S$, it holds that $|E(H)|=|E(H')|+\frac{|E(G,\{x_1,x_2\})|}{2}+1$, and so $\beta(G')+\frac{|E(G,\{x_1,x_2\})|}{2}+1\ge \beta(G)$. Conversely, by Lemma \ref{half}, $\beta(G)\ge \beta(G')+\frac{|E(G,\{x_1,x_2\})|}{2}+1$.
Finally, observe that $\pt{G}=\pt{G'}+\frac{|E(G,\{x_1,x_2\})|}{2}+1$, which implies that $\beta(G)-\pt{G}=\beta(G')-\pt{G'}$. Hence, $G$ admits a balanced subgraph of size $\pt{G}+\frac{k}{4}$ if and only if $G'$ admits a balanced subgraph of size $\pt{G'}+\frac{k}{4}$.

\smallskip

\noindent{\bf Rule \ref{kerRuleRemove1}:} Let $C,X$ and $x\in X$ be as in the description of Rule \ref{kerRuleRemove1}. By Lemma \ref{lem:partsizes}, there exists a $(V_1,V_2)$-balanced subgraph $H$ of $G$ with $\beta(G)$ edges, such that $|V_1\cap V(C)|=|V_2\cap V(C)|$. Consider the graph $G'=G-\{x\}$ formed by the application of the rule and the balanced subgraph $H'=H-\{x\}$. Then $|E(H)|=|E(H')|+\frac{|V(C)|}{2}$, and thus $\beta(G')\ge \beta(G)-\frac{|V(C)|}{2}$. Conversely, by Lemma \ref{half}, $\beta(G)\ge \beta(G')+\frac{|V(C)|}{2}$. However, $\pt{G}=\pt{G'}+\frac{|V(C)|}{2}-\frac{1}{4}$. Hence, $\beta(G)-\pt{G}= \beta(G')-\pt{G'}+\frac{1}{4}$.
Therefore, $G$ admits a balanced subgraph of size $\pt{G}+\frac{k}{4}$ if and only if $G'$ admits a balanced subgraph of size $\pt{G'}+\frac{k-1}{4}$.

\smallskip

\noindent{\bf Rule \ref{kerRsize3see2}:}
Let $C$ and $\{x,y,u\}$ be as in the description of Rule \ref{kerRsize3see2}. Firstly consider the case when $xy$ is a bridge in $G-\{u\}$. For any maximal balanced subgraph $H$ of $G$, without loss of generality one may assume that $xu,yu\in E(H)$ and $xy\notin E(H)$. Suppose $H$ is induced by a partition $(V_1,V_2)$ and $x,y\in V_1$. Form a balanced subgraph of $G'$ from $H-\{x,y,u\}$ by placing $z$ in $V_1$. Therefore, $\beta(G)=\beta(G')+2$. Since $\pt{G}=\pt{G'}+\frac{3}{2}+\frac{2}{4}=\pt{G'}+2$, it follows that $\beta(G)=\pt{G}+\frac{k}{4}$ if and only if $\beta(G')=\pt{G'}+\frac{k}{4}$.

Now consider the case when $xy$ is not a bridge in $G-\{u\}$. Then the graph $G'$ formed by deleting the vertex $u$ and the edge $xy$ is connected.
Furthermore, regardless of whether $x$ and $y$ are in the same partition that induces a balanced subgraph $H'$ of $G'$,
$H'$ can be extended to a balanced subgraph $H$ of $G$ such that $|E(H)|=|E(H')|+2$. This means that, as before, $\beta(G)=\beta(G')+2$.
But in this case $\pt{G}=\pt{G'}+\frac{7}{4}$ and thus $\beta(G)=\pt{G}+\frac{k}{4}$ if and only if
$\beta(G')=\pt{G'}+\frac{k-1}{4}$.

\smallskip

\noindent{\bf Rule \ref{kerRuleOneinS}:}
Let $T$ and $s\in S$ be as in the description of Rule \ref{kerRuleOneinS}. Since $\beta(G[V(T)\cup\{s\}])=\pt{G[V(T)\cup\{s\}]}+\frac{p}{4}$, by Lemma \ref{cutvertex}, $\beta(G)=\beta(G-T)+\pt{G[V(T)\cup\{s\}]}+\frac{p}{4}$. Also, by Lemma \ref{cutvertex}, $\pt{G}=\pt{G-T}+\pt{G[V(T)\cup\{s\}]}$. Hence $\beta(G)-\pt{G}=\beta(G-T)-\pt{G-T}+\frac{p}{4}$, which implies that $G$ admits a balanced subgraph of size $\pt{G}+\frac{k}{4}$ if and only if $G-T$ admits a balanced subgraph of size $\pt{G-T}+\frac{k-p}{4}$.
\end{proof}

%\subsection{Kernel Proof}

To show the existence of a kernel with $O(k^3)$ vertices, it is enough to give a bound on the number of non-path blocks, the number of vertices in these blocks and the number of path vertices. This is done by Corollaries \ref{cor:pathblock} and \ref{cor:boundexterior} and Lemma \ref{lem:blocksize}.

While Lemma \ref{lem:blocksize} applies to any graph reduced by Rule \ref{kerRuleRemove2}, the proofs of Corollaries \ref{cor:pathblock} and \ref{cor:boundexterior} rely on Lemma \ref{lem:BlockBound}, which gives a general structural result on forest of cliques with a bounded number of special blocks and bounded path length. Corollary \ref{cor:boundLeafs} and Lemma \ref{lem:pathlength} provide sufficient conditions for a reduced instance to be a {\sc Yes}-instance, thus producing a bound on the number of special blocks and the path length of {\sc No}-instances.
Lastly, Theorem \ref{thm:kernel} puts the results together to show the existence of the kernel.

Henceforth, we assume that the instance $(G,k)$ is such that $G$ is reduced by Rules \ref{kerRuleFirst}-\ref{kerRuleLast}, $G-S$ is a forest of cliques which does not contain a positive edge and $|S|\le 3k$. 

\begin{lemma}\label{lem:Snb}
Let $T$ be a connected component of $G-S$. Then for every leaf block
$C$ of $T$, $N_G(\interior{C})\cap S \neq \emptyset$. Furthermore,
if $|N_G(S)\cap V(T)|=1$, then $T$
consists of a single vertex.
\end{lemma}

\begin{proof}
We start by proving the first claim. Note that if $T=C$ consists of a
single vertex, then $N_G(\interior{C})\cap S \neq \emptyset$ since $G$ is
connected.
So assume that $C$ has at least two vertices. Suppose that
$N_G(\interior{C})\cap S =
\emptyset$ and let $X=\interior{C}$. Then if $|\interior{C}|>
|\exterior{C}|$, Rule
\ref{kerRuleRemove2} applies. If  $|\interior{C}|=|\exterior{C}|$ then
Rule \ref{kerRuleRemove1} applies. Otherwise,
$|\interior{C}|<|\exterior{C}|$ and since $|\exterior{C}|\le 1$ (as $C$ is
a leaf block), $C$ has only one vertex, which contradicts our assumption
above.
For the second claim,
first note that since $|N_G(S)\cap V(T)|=1$, $Q$ has one leaf block and so
$T$ consists of a single block.
Let $N_G(S)\cap V(T)= \{ v \}$ and $X=V(T)-\{v\}$. If
$|X|>1$, Rule \ref{kerRuleRemove2} applies. If $|X|=1$, Rule
\ref{kerRuleRemove1} applies. Hence $V(T)=\{v\}$.

\end{proof}

Let ${\cal B}$ be the set of non-path blocks.

\begin{lemma}\label{lem:NhoodS}
If there exists a vertex $s\in S$ such that $\sum_{C\in {\cal B}} |N_G(\interior{C})\cap \{s\}| \ge 2(|S|-1+k)$, then $(G,k)$ is a {\sc Yes}-instance.
\end{lemma}

\begin{proof}
Form $T\subseteq N_G(s)$ by picking a vertex from each block $C$
for which $|N_G(\interior{C})\cap \{s\}|=1$: if there exists a vertex
$x\in\interior{C}$ such that $N_G(x)\cap S=\{s\}$, pick this,
otherwise pick $x\in \interior{C}$ arbitrarily. Let $U=T\cup \{s\}$
and $W=V\setminus U$.

Observe that $G[U]$ is balanced by Theorem \ref{Htheorem} as $G[U]$ is
a tree. Thus $\beta(G[U])=|T|=\frac{|T|}{2}+\frac{|T|}{4}+\frac{|T|}{4}=\pt{G[U]}+\frac{|T|}{4}$.

%Next observe that $G[W]$ has at most $(|S|-1)+\frac{|T|}{2}$ connected components. 
Consider a connected component $Q$ of $G-S$. By Rule
\ref{kerRuleOneinS}, $|N_G(Q)\cap S|\ge 2$ and by Lemma \ref{lem:Snb}, if
$|N_G(S)\cap V(Q)|=1$ then $Q$ consists of a single vertex. Otherwise,
either $(N_G(S) \setminus N_G(s)) \cap V(Q) \neq \emptyset$, or $Q$ has
at least two vertices in $T$. Moreover, note that the
removal of interior vertices does not disconnect the component itself.
Hence $G[W]$ has at most
$(|S|-1)+\frac{|T|}{2}$ connected components. 
Applying Lemma \ref{half},
$\beta(G)\ge\pt{G}+\frac{|T|}{4}-\frac{(|S|-1)+\frac{|T|}{2}}{4}=\pt{G}+\frac{|T|}{8}-\frac{|S|-1}{4}$.
Hence if $|T|\ge 2(|S|-1+k)$, then $(G,k)$ is a {\sc Yes}-instance.
\end{proof}

\begin{corollary}\label{cor:boundLeafs}
If $\sum_{C\in {\cal B}} |N_G(\interior{C})\cap S| \ge |S|(2|S|-3+2k)+1$, the instance is a {\sc Yes}-instance. Otherwise, $\sum_{C\in {\cal B}} |N_G(\interior{C})\cap S| \le  3k(8k-3)$.
\end{corollary}
\begin{proof}
If $\sum_{C\in {\cal B}} |N_G(\interior{C})\cap S| \ge |S|(2|S|-3+2k)+1$, then for some $s\in S$ we have $\sum_{C\in {\cal B}} |N_G(\interior{C})\cap \{s\}| \ge 2|S|-3+2k+1/|S|$ and, since the sum is integral, $\sum_{C\in {\cal B}} |N_G(\interior{C})\cap \{s\}| \ge 2(|S|-1+k)$. Thus, $(G,k)$, by Lemma \ref{lem:NhoodS}, is a {\sc Yes}-instance. The second inequality of the corollary follows from the fact that $|S|\le 3k.$
\end{proof}

\begin{lemma}\label{lem:pathlength}
If in $G-S$ there exist vertices $U=\{u_1,u_2,\ldots,u_p\}$ such that  $N_{G-S}(u_i)= \{u_{i-1},u_{i+1} \}$ for $2\leq i\leq p-1$, and $p\ge |S|+k+1$, then $(G,k)$ is {\sc Yes}-instance. Otherwise, $p\leq 4k$.
\end{lemma}
\begin{proof}
Observe that $G[U]$ is balanced by Theorem \ref{Htheorem}. Thus $\beta(G[U])=p-1=\pt{G[U]}+\frac{p-1}{4}$. 
Let $W=V\setminus U$ and observe that $G[W]$ has at most $|S|$ components, since, by Lemma  \ref{lem:Snb}, every vertex in $G-U$ has a path to a vertex in $S$.
Applying Lemma \ref{half}, $\beta(G)\ge\pt{G}+\frac{p-1}{4}-\frac{|S|}{4}$. Hence if $p-1-|S|\ge k$, $(G,k)$ is a {\sc Yes}-instance.
\end{proof}

\begin{lemma}\label{pathorspecial}
A block $C$ in $G-S$ such that $|\exterior{C}|=2$ is either special or it is a path block.
\end{lemma}

\begin{proof}
Suppose $C$ is not special. If $|V(C)|\ge 5$, then Reduction Rule \ref{kerRuleRemove2} would apply. 
If $|V(C)|= 4$, then Reduction Rule \ref{kerRuleRemove1} would apply. If $|V(C)|= 3$, then Reduction Rule \ref{kerRsize3see2} would apply. Hence $|V(C)|=2$ and it is a path block.
\end{proof}

In $G-S$, a {\em pure path} is a path consisting exclusively of path vertices. Note that every path vertex belongs to a unique pure path.

\begin{lemma}\label{lem:BlockBound}
Suppose $G-S$ has at most $l$ special blocks and the number of vertices in each pure path is bounded by $p$. 
Then $G-S$ contains at most $2l$ non-path blocks and $2pl$ path vertices.
\end{lemma}
\begin{proof}
It suffices to prove that if every connected component $T$ of $G-S$ has at most $l_T$ special blocks, then $T$ contains at most $2l_T$ non-path blocks and $2pl_T$ path vertices. So, we may assume that $T=G-S$ is connected. Pick an arbitrary non-path block $C_R$ as the `root' node. Define the distance $d(C_R,C)$ as the number of non-path blocks different from $C_R$ visited in a path from a vertex in $C_R$ to a vertex in $C$. For every non-path block $C$ in $T$, the {\em parent} block $C'$ is the unique non-path block such that $C'$ contains an edge of any path from $C_R$ to $C$ and $d(C_R,C)-d(C_R,C')=1$. In addition, $C_R$ is the parent of every block $C$ such that $d(C_R,C)=1$.

Consider the tree $F$ that contains a vertex for every non-path block of $T$ and such that there is an edge between two vertices if and only if one of the corresponding blocks is the parent of the other. Observe that given a vertex $v\in F$ which corresponds to a block $C$ of $T$, it holds that $d_F(v)\geq |\exterior{C}|$. In addition, by Lemma \ref{lem:Snb}, every leaf in $F$ corresponds to a special block.

Now, we know that in a tree the number of vertices of degree greater or equal to three is bounded by the number of leaves. Moreover, by Lemma \ref{pathorspecial}, if a block $C$ is such that $|\exterior{C}|=2$, then it is either special or a path block. Thus, the number of non-path blocks is bounded by $2l$.

Furthermore, note that the number of pure paths in $T$ is bounded by the number of edges in $F$, which is bounded by $2l-1$. Since every pure path contains at most $p$ path vertices, the number of path vertices is bounded by $(2l-1)p< 2pl$.
\end{proof}

\begin{corollary}\label{cor:pathblock}
$G-S$ contains at most $6k(8k-3)$ non-path blocks and $24k^2(8k-3)$ path vertices.
\end{corollary}
\begin{proof}
By Corollary \ref{cor:boundLeafs}, $G-S$ contains at most $3k(8k-3)$ special blocks and by Lemma \ref{lem:pathlength}, the length of every pure path is bounded by $4k$. Thus, Lemma \ref{lem:BlockBound} implies that $G-S$ contains at most $6k(8k-3)$ non-path blocks and $24k^2(8k-3)$ path vertices.
\end{proof}

\begin{corollary}\label{cor:boundexterior}
$G-S$ contains at most $12k(8k-3)$ vertices in the exteriors of non-path blocks.
\end{corollary}
\begin{proof}
For any component $T$ of $G-S$, consider the tree $F$ defined in the proof of Lemma \ref{lem:BlockBound}. For any block $C$ of $T$ and any vertex $v$ in $\exterior{C}$, $v$ corresponds to an edge of $F$. Furthermore, for any edge of $F$ there are at most two exterior vertices in $T$ that correspond to it. Therefore, $|\cup_{C\in {\cal B}}\exterior{C}|\leq 2|{\cal B}|\leq 12k(8k-3)$.
\end{proof}

\begin{lemma}\label{lem:blocksize}
For a block $C$, if $|V(C)|\ge 2|\exterior{C}|+|N_G(\interior{C})\cap S|
(2|S|+2k+1)$, then $(G,k)$ is a {\sc Yes}-instance. Otherwise, $|V(C)|\le
2|\exterior{C}|+|N_G(\interior{C})\cap S| (8k+1)$.
\end{lemma}
\begin{proof}
Consider a fixed $s\in N_G(\interior{C})\cap S$. We will show that we
may assume that either $|N_G^+(s)\cap \interior{C}|\le \frac{k+|S|}{2}$ or
$|N_G^+(s)\cap
\interior{C}|\ge
|\interior{C}|-\frac{k+|S|}{2}$, because otherwise $(G,k)$ is a {\sc
Yes}-instance.

Indeed, suppose $\lceil\frac{k+|S|}{2}\rceil \leq |N_G^+(s)\cap
\interior{C}|\leq|\interior{C}|-\lceil\frac{k+|S|}{2}\rceil$. Let
$U_1\subseteq N_G^+(s)\cap \interior{C}$, $|U_1|=\lceil
\frac{k+|S|}{2} \rceil$, and let $U_2\subseteq\interior{C}\setminus
N_G^+(s)$, $|U_2|=\lceil \frac{k+|S|}{2} \rceil$.
Let $U=U_1\cup U_2\cup\{s\}$ and consider the subgraph $H$ of $G[U]$
induced by the edges $E(U_1,U_2)\cup E(s,(U_1\cap N_G^+(s)))\cup
E(s,(U_2\cap N_G^-(s))).$ Observe that $H$ is
$(U_1\cup\{s\},U_2)$-balanced and so $\beta(G[U])\ge |U_1|^2+|U_1\cap
N_G^+(s)|+|U_2\cap N_G^-(s)|$. Furthermore,
$\pt{G[U]}=|U_1|^2+\frac{|U_1\cap N_G^+(s)|}{2}+\frac{|U_2\cap
N_G^-(s)|}{2}$, and
hence $\beta(G[U])\ge \pt{G[U]}+\frac{|U_1\cap N_G^+(s)|+|U_2\cap
N_G^-(s)|}{2}\geq\pt{G[U]}+\frac{k+|S|}{4}$.

Now consider $W=V\setminus U$. Any connected component of $G-S$ is
connected to two vertices in $S$, hence $G[W]$ has at most $|S|-1$
components adjacent to vertices in $S\setminus \{s \}$ and one component
corresponding to the block $C$.
Applying Lemma \ref{half}, $\beta(G)\ge\pt{G}+\frac{(k+|S|)-|S|}{4}$,
which means that $(G,k)$ is a {\sc Yes}-instance.

Similarly, we can show that we may assume that either $|N_G^-(s)\cap
\interior{C}|\le \frac{k+|S|}{2}$ or $|N_G^-(s)\cap
\interior{C}|\ge
|\interior{C}|-\frac{k+|S|}{2}$, because otherwise $(G,k)$ is a {\sc
Yes}-instance.

Let $S_1^+=\{s\in S: 0 < |N_G^+(s)\cap\interior{C}| \leq
\frac{k+|S|}{2}\}$, $S_2^+=(N_G^+(\interior{C})\cap S)\setminus S_1^+$ and
$X^+=\{v\in\interior{C}\setminus N_G^+(S_1^+):v\in N_G^+(s),\forall s\in
S_2^+\}$. 
%Informally speaking, $X^+$ is the set of vertices of
%$\interior{C}$ which have positive edges towards every vertex in $S_2^+$
%and no edges towards $S_1^+$ (and the rest of $S$).
Observe that for all $s\in S_2^+$, $|N_G^+(s)\cap\interior{C}|\geq
|\interior{C}|-\frac{k+|S|}{2}$, which means that $|X^+|\geq
|\interior{C}\setminus N_G^+(S_1^+)|-|S_2^+|\frac{k+|S|}{2}$. In addition,
$|N_G^+(S_1^+)\cap\interior{C}|\leq |S_1^+|\frac{k+|S|}{2}$, hence
$|\interior{C}\setminus N_G^+(S_1^+)|\geq
|\interior{C}|-|S_1^+|\frac{k+|S|}{2}$. Therefore, $|X^+|\geq
|\interior{C}|-(|S_1^+|+|S_2^+|)\frac{k+|S|}{2}=|\interior{C}|-|N_G^+(\interior{C})\cap
S|\frac{k+|S|}{2}\geq |\interior{C}|-|N_G(\interior{C})\cap
S|\frac{k+|S|}{2}$.

With similar definitions and the same argument we obtain $|X^-|\geq
|\interior{C}|-|N_G(\interior{C})\cap S|\frac{k+|S|}{2}$. Now let $X=X^+\cap X^-$ 
and observe
that and $|X|\geq
|\interior{C}|-|N_G(\interior{C})\cap S|(k+|S|)$.

%Hence, we may assume that there is a set $X\subseteq\interior{C}$ of
%vertices with the same positive and negative neighborhoods in $S$
%such that $|X|\ge |\interior{C}|-|N_G(\interior{C})\cap S|(|S|+k)$.
However, by Rule \ref{kerRuleRemove2},
$|X|\leq\frac{|V(C)|+|N_G(\interior{C})\cap S|}{2}$.
So, $|\interior{C}|\le |N_G(\interior{C})\cap S|
(|S|+k+\frac{1}{2})+\frac{|V(C)|}{2}$, and so $|V(C)|\le
2|\exterior{C}|+|N_G(\interior{C})\cap S| (2|S|+2k+1)$ as claimed.
\end{proof}

\begin{theorem}\label{thm:kernel}
{\smc} has a kernel with $O(k^3)$ vertices.
\end{theorem}
\begin{proof}
Let $(G=(V,E),k)$ be an instance of \smc. As in Theorem \ref{thm:fpt}, apply Rules \ref{rule:first}-\ref{rule:last} exhaustively: either the instance is a {\sc Yes}-instance, or there exists $S\subseteq V$ such that $|S|\leq 3k$ and $G-S$ is a forest of cliques which does not contain a positive edge.

Now, apply Rules \ref{kerRuleFirst}--\ref{kerRuleLast} exhaustively to $(G,k)$ to obtain a new instance $(G',k')$. If $k'\leq 0$, then $(G,k)$ is a {\sc Yes}-instance since Rules \ref{kerRuleFirst}--\ref{kerRuleLast} are valid. Now let $G=G', k=k'$.
Check whether $(G,k)$ is a {\sc Yes}-instance due to Corollary \ref{cor:boundLeafs}, Lemma \ref{lem:pathlength} or Lemma \ref{lem:blocksize}. If this is not the case, by Corollary \ref{cor:pathblock}, $G-S$ contains at most $6k(8k-3)$ non-path blocks and $24k^2(8k-3)$ path vertices. Hence, by Lemma \ref{lem:blocksize}, $|V(G)|$ is at most
$$
|S|+ 24k^2(8k-3)+ \sum_{C\in {\cal B}} |V(C)|
\le O(k^3)+ 2\sum_{C\in {\cal B}}|\exterior{C}|+(8k+1)\sum_{C\in {\cal B}}|N_G(\interior{C})\cap S|
$$

Now, applying Corollary \ref{cor:boundLeafs} and Corollary \ref{cor:boundexterior}, we obtain:
$$
|V(G)| \le O(k^3) +  48k(8k-3)+3k(8k-3)(8k+1)=O(k^3).
$$

\end{proof}

It is not hard to verify that no reduction rule of this paper increases the number of positive edges. Thus, considering an input $G$ of {\sc Max Cut ATLB} as an input of  {\sc Signed Max Cut ATLB} by assigning minus to each edge of $G$, we have the following: 

\begin{corollary}\label{cor:MC}
{\sc Max Cut ATLB} has a kernel with $O(k^3)$ vertices.
\end{corollary}

\section{Extensions and Open Questions}

In the previous sections, the input of {\sc Signed Max Cut ATLB} is a
signed graph without parallel edges. However, in some applications
(cf. \cite{GGMZ,GZ}), signed graphs may have parallel edges of opposite
signs.
We may easily extend inputs of {\sc Signed Max Cut ATLB} to such
graphs. Indeed, if $G$ is such a graph we may remove all pairs of
parallel edges from $G$ and obtain an equivalent instance of {\sc
Signed Max Cut ATLB}.

In fact,  the Poljak-Turz\'{i}k bound can be extended to edge-weighted
graphs \cite{PoljakTurzik86}. Let $G$ be a connected  signed graph in
which each edge $e$ is assigned a positive weight $w(e)$. The weight
$w(Q)$ of an edge-weighted graph $Q$ is the sum of weights of its
edges. Then $G$ contains a balanced subgraph with weight at least
$w(G)/2+w(T)/4,$ where $T$ is a spanning tree of $G$ of minimum
weight \cite{PoljakTurzik86}.
It would be interesting to establish parameterized complexities of
{\sc Max Cut ATLB} and {\sc Signed Max Cut ATLB} extended to
edge-weighted graphs using the Poljak-Turz\'{i}k bound above.

\vspace{2mm}

%\noindent{\bf Acknowledgement.}  We are thankful to Fedor Fomin for a useful discussion.

\end{document}